\documentclass[]{lipics-v2021}
\pdfoutput=1
\nolinenumbers

\usepackage{multicol}

\usepackage{mathtools}
\usepackage{eqnarray} 
\usepackage{mathpartir}

\def\lsq{\text{\texttt{\upshape [}}}
\def\rsq{\text{\texttt{\upshape ]\hspace{-1.5pt}}}}
\def\li[#1]{\lsq#1\rsq}
\def\kw#1{\mathsf{#1}}
\def\kwlet{\mathop{\kw{let}}}
\def\kwin{\mathop{\kw{in}}}
\def\be{\mathop{\Leftarrow}}
\def\Red#1{\kw{Red}_{#1}}
\def\dotsucc{\mathrel{\dot\succ}}

\def\hole{ {[\hskip 0.5em]} }
\def\rwml{\rw_{ml}}

\def\rw{\leadsto}

\def\id#1{\mathit{#1}}

\def\effpure{\mathit{pure}}

\def\lml{\lambda_{ml}}
\def\lc{\lambda_{c}}
\def\case{\item \textsc{case} }

\title{The Marriage of Effects and Rewrites}
\author{Ezra e.\ k.\ Cooper}{afilliation}{ezra@ezrakilty.net}{orcid}{independent}
\authorrunning{E.\ E.\ K.\ Cooper}
\Copyright{Ezra e.\ k.\ Cooper}

\begin{document}

\maketitle

\begin{abstract}
In the research on computational effects, defined algebraically, effect symbols are often expected to obey certain equations. If we orient these equations, we get a rewrite system, which may be an effective way of transforming or optimizing the effects in a program. In order to do so, we need to establish strong normalization, or termination, of the rewrite system. Here we define a framework for carrying out such proofs, and extend the well-known Recursive Path Ordering of Dershowitz to show termination of some effect systems.
\end{abstract}

\section{Introduction}

Plotkin and Power~\cite{plotkin2001adequacy} introduced a view on computational effects as algebraic terms. Their operational semantics shows how a source term containing effects within it can reduce to a final effect term that represents a trace of the effects performed by the program, or indeed a tree of all possible linear traces. That line of work discusses equations between effect operators, which define their essence in relation to one another.

We add to that work with an observation: by orienting the equations as a system of \emph{rewrite rules} for the effect system, we can mechanically reduce an effect term (as a \emph{trace}) to something more compact, a \emph{state}. The rewrite rules in this case can be thought of as taking the place of the language implementation or indeed the hardware which makes the effect ``take effect.'' Alternatively, such rewrite rules can be an elegant way of defining and implementing optimizations for effectful programs.

A general rewrite system is not the only way to work from an effectful term or trace to a final state. Ahman and Staton~\cite{ahman2013normalization} give a Normalization By Evaluation strategy for doing just that. However, we find it interesting to study the behavior of generalized rewrite systems. The freedom to apply rewrite rules arbitrarily could be useful in a compiler implementation or other program-transformation engine.

Having applied general term-rewriting to algebraic effects, the researcher will want to know whether common properties apply, such as confluence and strong normalization. While an ad-hoc proof may be given, it would be preferable to factor the problem so that a proof about the termination of the term-rewrite system alone would easily lift to a proof about the system in the context of Moggi's computational metalanguage \cite{moggi1988computational}.

Now, in the literature of term-rewriting there are many techniques for showing strong normalization, for example the ``recursive path ordering'' or RPO~\cite{dershowitz1982orderings}. With this technique, the practitioner merely exhibits an ordering on the function-symbols and shows that each rewrite rule obeys this ordering in a certain way. Thus an intricate inductive proof is replaced by some relatively simple (albeit recursive) checks on the rewrite rules. The RPO is extended to a calculus with $\lambda$-abstraction and with $\beta$-reduction in the literature on the ``higher-order recursive path ordering,'' HORPO~\cite{jouannaud1999higher}. But that leaves us to wonder, what about a calculus with the let-construct? If we can extend RPO or HORPO to the computational metalanguage, then we will have an easy way of proving termination of some systems of effects.

In this paper, we make the following contributions:
\begin{itemize}
\item We show how to interpret a variety of computational effects operationally as rewrite-rules, rather than equations, giving a more operational flavor to the workings of effect symbols,
\item We show a new technique for proving strong normalization for rewriting algebraic-effect systems by lifting the recursive path ordering into Moggi's metalanguage.
\item We use the technique to prove termination of an effect system for global state, which shows how to reduce its \emph{traces} to \emph{states}.
 \end{itemize}

Our central proof will not be surprising to anyone familiar with RPO, with the Tait-Girard proof of strong normalization~\cite{girard1989proofs} and the Lindley-Stark method for extending Tait-Girard to the computational metalanguage~\cite{lindley2005reducibility}. But by combining all these things, we get a compelling result that can be used directly to show termination of languages in the presence of algebraic effects, where the effects themselves are subject to rewrite rules.

\section{Algebraic effects}

Let's review the basic framework of algebraic effects, introduced by Plotkin and Power~\cite{plotkin2001adequacy} (and Bauer~\cite{bauer2019whatsalgebraic}) as extended in Plotkin and Pretnar~\cite{plotkin2008logic} with the $\kwlet$-construct. A programming language is defined with effect symbols representing individual atomic effects that can be performed. The symbols build on two kinds of syntactic roles, parameters $p$ and arguments $a$, as $e_{\vec p}(\vec a)$. The parameters represent data that is used by the effect, such as a message to print to the terminal, while the arguments represent possible continuation terms, depending on the \emph{result} of the effect. An effect $readbit$ which reads a single bit from some input source would naturally have two argument positions, representing the behavior the program will follow if it reads a 0 or a 1, respectively: $readbit(zeroContinuation, oneContinuation)$.

Another effect, $print_m$ could be used to print a corresponding message: \[readbit(print_{\mathtt{cold}}(30),\, print_{\mathtt{hot}}(70))\] will print cold or hot correspondingly, then return a number, say 30 or 70 for a number of degrees celsius. Any pure (effect-free) term placed as an effect argument represents an ultimate return \emph{value} of the computation, dependent upon the path taken from root to leaf.

The Plotkin and Power semantics lets the effect symbols commute out of evaluation contexts, essentially floating to the top of the term at evaluation time, so that the normal forms are trees of effect symbols, representing possible traces, whose leaves are the corresponding return values. Thus we could apply the lambda-term $\lambda x. x+5$ to the above effectful term and it would rewrite as follows:
\begin{eqnarray*}
&& (\lambda x. x+5) readbit(print_{\mathtt{cold}}(30), print_{\mathtt{hot}}(70)) \\
&\rw& readbit((\lambda x. x+5) print_{\mathtt{cold}}(30), (\lambda x. x+5) print_{\mathtt{hot}}(70)) \hskip 5em & (eff-assoc) \\
&\rw& readbit(print_{\mathtt{cold}}((\lambda x. x+5) \: 30), print_{\mathtt{hot}}((\lambda x. x+5)\: 70)) \hskip 5em & (eff-assoc) \\
&\rw& readbit(print_{\mathtt{cold}}(30+5), print_{\mathtt{hot}}(70+5)\hskip 5em  & (abs-$\beta$) \\
&\rw& readbit(print_{\mathtt{cold}}(35), print_{\mathtt{hot}}(75) & (abs-$\beta$)
\end{eqnarray*}
The final row is a normal form, which shows a tree where the $readbit$ operation can choose either of two paths; on each path some specific message is printed; and finally each path terminates in a value, which was computed by applying the $\lambda$-abstraction to the \emph{result} of the original side-effecting expression.

In the present work, we explore what happens when these computation (effect) trees are further exposed to their own rewrite rules, which can be applied in source terms or in these final computation trees. Such reduction loses the ``trace'' nature of the tree, but gives us a way of simulating the machinery of the effects, something like an abstract machine for effects.

We do not make use of the effect parameters in our proofs, so we do not write them.

\subsection{Examples}

When we are talking about the semantics of a programming language, the rewrite rules for the effects can be seen as implementing the machinery of the language implementation which reduces the individual effects to a state, itself represented as a normal form of an effect-term.

\paragraph*{Example: Global State}

Global state is modeled as a single global location which can hold a value of some type $T$. The signature of the global-state effect system is\
\begin{align*}
    &arity(assign_i) = 1 \\
    &arity(get) = T
\end{align*}
Note $assign_i$ is parameterized by the value $i$, to assign into the global variable. Plotkin and Power distinguish ``parameters'' and ``arguments''. The arguments of $get$ are indexed by the values of the storage type $T$: its ``arity'' is $T$.

If the symbols are uninterpreted (subject to no rewrites) then the result of a rewrite sequence is just a computation tree, which acts as a tree of all possible traces of the program. But we may alternatively assign a meaning which is the actual final state of this computation, in cases where there is one. To that end, we can assign rewrite rules (adapted from Plotkin and Power~\cite{plotkin2002notions}) that perform the trace-reduction:
\begin{align*}
  assign_i(get(t_1,\, \dots,\, t_n)) &\rw assign_i(t_i)\\
  assign_i(assign_j(s)) &\rw assign_j(s) \\
  get(t_1,\, \dots,\, t_i,\, \dots,\, t_n) &\rw get(t_1,\, \dots,\, s_i,\, \dots,\, t_n)\\
&\text{
\hspace{2em}
where $t_i = get(s_1,\, \dots,\, s_n)$}
\end{align*}
The normal forms of this system have no adjacent $\id{get}$-$\id{get}$ pairs, no adjacent $\id{assign}$-$\id{assign}$ pairs, and no $\id{get}$ inside an $\id{assign}$, so they are really just a single $\id{assign}$, or a single $\id{get}$, with argument as the final value result of the program. The single $\id{assign}$ indicates to us what the final state of the global variable was, by its \emph{parameter}. A single $\id{get}$ would represent a read of an uninitialized variable, or a nondeterministic one; in a particular setting this might be disallowed by other mechanisms which don't interest us here.

Contrast this rewriting approach with the examples given in Johann, et al.~\cite{johann2010generic} where global-state computation-trees are mapped to their final state by a function defined outside the calculus.

\paragraph*{Example: Nondeterminism}

There is one effect, $or$, with $arity(or) = 2$, and one rewrite rule:
\[
or(or(s_1, s_2), s_3) \rw or(s_1, or(s_2, s_3)) 
\]
The rule has the purpose of normalizing a branching tree of possible computations to a flat list of possible outcomes, and so works more like a list.

\paragraph*{Example: a looping effect}

This effect has not been proposed in the literature to our knowledge, but to motivate our work, we explore the idea of something that looks like an ``effect'' but has some complex rewriting behavior.

In the practice of programming with external services (for example, database servers, or web-based APIs), one frequently wants to make one's own service robust in the face of a brief interruption to the external service. To that end, the programmer builds a finite number of retries into their system. If the external service begins functioning during the retries, the program will continue normally, but if the finite retries are exhausted, an error is returned to the user.

We model such a system using a pair of effects, request and retry. Each time the program makes a \emph{request} to the service, that request may fail (a possibility whose continuation is represented by a first parameter, $t$), or it may return a meaningful value (represented by an indexed set of parameters, $s_1,\, \dots,\, s_n$). So we introduce an effect $request(t,\, s_1,\, \dots,\, s_n)$. We also introduce an effect $retry(u,\, r)$ which represents the effect of retrying the computation $r$ a number of times indicated by $u$. In $u$ we will find a number represented through rewrite symbols $zero()$ and $succ(u)$, i.e.\ Peano numerals (we use Peano numerals to make the arithmetic amenable to rewriting).
\begin{eqnarray*}
retry(zero(), request(t, s_1,\, \dots,\, s_n)) &\rw& t \\
retry(succ(u), request(t, s_1,\, \dots,\, s_n)) &\rw& request(retry(u, t'), s_1,\, \dots,\, s_n)\\
&&&where $t' = request(t, s_1,\, \dots,\, s_n)$
\end{eqnarray*}
This effect-rewrite system produces something more like a trace than a final state, since it replicates the $request$ effect $u$ times in the computation tree. To evaluate the trace, we could choose further rewrite rules that make $request$ act like $get$ in the global-state example, flattening successive $request$s and choosing a single outcome for the whole set.

This may be a contrived model for a retry-loop of effects, but it demonstrates our technique on a slightly more complex system than the other examples.

\paragraph*{Example: Parallelism}

A binary effect, $\id{par}$, represents parallel evaluation of two streams of effects. We assume it is used in combination with other effects. The following rewrite rules are replicated for each other effect $e$ in the system:
\begin{align*}
par(e(s_1,\, \dots,\, s_n), t) \rw e(par(s_1, t),\, \dots,\, par(s_n, t))\\
par(s, e(t_1,\, \dots,\, t_n)) \rw e(par(s, t_1),\, \dots,\, par(s, t_n))
\end{align*}
These rules are not in general confluent, so several different final states can be derived from a single source term. That is of course in the nature of parallelism.

In the tradition of fork-join parallelism, we could also add an effect $\id{join}$ which brings together the two results in one result term, for further computation. Here $\langle \cdot, \cdot \rangle$ represents ordinary data pairing into a product type ($S \times T$):
\[join(par(v, w)) \rw \id{pure}(\langle v, w \rangle) \]

\section{Two kinds of effectful metalanguage}

We must pause to reconcile two syntactic treatments of effects in the literature. One marks the monad type explicitly, the other leaves it implicit. Both treatments appear in Moggi's early work and in the literature are often referred to as $\lml$ and $\lc$.

The first approach ($\lml$) uses a computation type while the other ($\lc$) treats effect operations as transparent to the type system. The latter notation predominates in Plotkin and Power~\cite{plotkin2001adequacy} and other algebraic-effects research. Sabry and Wadler~\cite{ReflectionOnCBV} establish a close correspondence between them. 

The key typing rule for each is given below. We write $\kw{E}(T)$ for the type of an effectful computation giving result type $T$. (A single such effect constructor implies one global monad for effects throughout the system.)

\begin{multicols}{2}
Explicit computation types ($\lml$)

\inferrule{\Gamma \vdash t : \kw{E}(S) \\ \Gamma, x : S \vdash u : \kw{E}(T) }
          {\Gamma \vdash \kwlet x \be t \kwin u : \kw{E}(T)}

\columnbreak
Effects as type-transparent ($\lc$)

\inferrule{\Gamma \vdash t : S \\ \Gamma, x : S \vdash u : T}
          {\Gamma \vdash \kwlet x \be s \kwin u : T}
\end{multicols}

In $\lml$, we perform beta-reduction on lets with explicitly-constructed pure subjects:
\[\kwlet x \be pure(t) \kwin u \rw u\{t/x\}\]
In $\lc$, beta-reduction is triggered by the syntactic class of a value in the subject position (assume $v$ describes a syntactic class of values):
\[\kwlet x \be v\kwin u \rw u\{v/x\}\]
Values $v$ in $\lc$ are defined by a grammar, which prohibits $\kw{let}$ and effect application, at least when not embedded in a $\lambda$-body.

The main proofs in this paper use $\lml$ as the substrate.

\section{A Metalanguage With Explicit Effects}

Now we define our core object: a metalanguage for computational effects, based on the basic syntax of Moggi~\cite{moggi1988computational} with the algebraic-effect rules of Plotkin and Pretnar~\cite{plotkin2008logic}. As a blend of those languages, it includes explicit effect symbols and a let-construct.

Unlike some later work on algebraic effects, we don't use the \emph{fine-grain call-by-value} of Levy, et al.~\cite{levy2003modelling}, and in fact do not assume a call-by-value evaluation order, because we want to cast as wide a net as possible for the interesting rewrite systems that can be proven strongly-normalizing with our technique. In fact, one of our motivating examples (NRC) benefits from allowing rewrites in arbitrary position and from allowing arbitrary term-term applications (in distinction to FGCBV).

Here is the grammar of our metalanguage:
\begin{eqnarray*}
s, t, u &::=& x \mid \lambda x. u \mid \effpure(t) \mid st \mid \gamma(t_1, ..., t_n) \mid \kwlet x \be t \kwin u\\
\gamma &::=& e \mid f & (rewritable symbols)\\
e, e', e'' &&& \text{(effect symbols)}\\
f, g && &\text{(function symbols)}%
\end{eqnarray*}
We distinguish two classes of rewritable symbols: the effect symbols and function symbols. These two classes have their own typing rules, but are often treated the same in the rewrite theory, so we use $\gamma$ to range over both, $e$ (and its primes) to range over effect symbols, and $f, g, ...$ to range over non-effect function symbols. Both subclasses are subject to some metalanguage rewrites as motivated by the Plotkin-Power framework.

Types are defined by this grammar:
\begin{eqnarray*}
S,\, T &::=& B \mid \kw{E}(T) \mid S \to T\\
B &&\text{(basic types)}
\end{eqnarray*}%
The language has a special type constructor for effect types: $\kw{E}(T)$ is the type of computations that may be effectful, returning a value of type $T$. The type of functional abstractions is $S\to T$.

Typing rules are given in Figure~\ref{fig:typing-rules}. To give types to the rewritable symbols, we assume a signature $\Sigma$ which maps each effect symbol to an arity and each function symbol to a type signature of the form $S_1 \times \cdots \times S_n \to T$. Throughout the paper we will use series of terms like $s_1,\,\dots,\,s_n$ or a vector notation $\vec s$ interchangeably. There may be a distinct term called $s$ in such a context.

\begin{figure}
\caption{Typing rules}
\begin{mathpar}
\inferrule{\Gamma, x : S \vdash u : T}
        {\Gamma \vdash \lambda x. u : S \to T}
\and
\inferrule{\Gamma \vdash  s : S \to T \\ \Gamma \vdash  t : S}
         {\Gamma \vdash st : T}
\and
\inferrule{\Gamma \vdash t : \kw{E}(S) \\ \Gamma, x : S \vdash u : \kw{E}(T) }
          {\Gamma \vdash \kwlet x \be t \kwin u : \kw{E}(T)}
\and
\inferrule{\Gamma \vdash  t : T}
          {\Gamma \vdash  \effpure(t) : \kw{E}(T)}
\and
\inferrule{\text{for each $t_i$,\;} \Gamma \vdash  t_i : \kw{E}(T) \\ \Sigma(e) = n}
          {\Gamma \vdash  e(t_1,\, \dots,\, t_n) : \kw{E}(T)}
\and
\inferrule{\text{for each $s_i$,\;} \Gamma \vdash  s_i : S_i \\ \Sigma(f) = S_1 \times \cdots \times S_n \to T}
          {\Gamma \vdash  f(s_1,\, \dots,\, s_n) : T}
\end{mathpar}
\label{fig:typing-rules}
\end{figure}

\begin{figure}
\caption{Rewrite rules of the metalanguage.}
\begin{eqnarray*}
(\lambda x. u)t &\rwml& t\{u/x\} & (abs-$\beta$)\\
\kwlet x \be \effpure(t) \kwin u &\rwml& t\{u/x\}  & (let-$\beta$)\\
\kwlet y = (\kwlet x \be t_1 \kwin t_2) \kwin u & \rwml&
    \kwlet x \be t_1 \kwin \kwlet y \be t_2 \kwin u\hspace{2em}  & when $x\not\in FV(u)$\\
    &&& (let-assoc)\\
\kwlet x \be e(t_1,\, ...,\, t_n) \kwin u &\rwml& e(t_1',\, ...,\, t_n') & (eff-assoc)\\
  &&\hbox{\quad where $t_i' = \kwlet x \be t_i \kwin u$}%
\end{eqnarray*}
\label{fig:rewrite-ml}
\end{figure}
Rewrite rules for the metalanguage are given in Figure~\ref{fig:rewrite-ml}.
The eff-assoc rule is adapted from a similar equation in Plotkin and Pretnar~\cite{plotkin2008logic}. Significantly, eff-assoc does \emph{not} apply to function symbols, only effect symbols. We omit eta-reduction, which is found in Moggi~\cite{moggi1988computational} and in much of the other work, but it doesn't serve our proof or our examples.

Having prepared the metalanguage as a broth, the soup will be made by adding further rewrite rules, all algebraic in nature, and specific to a domain area, as shown in the Examples. So, assume we are given an underlying rewrite system, defined by a signature $\Sigma$ of rewritable symbols with arities and their rewrite rules $l \rw r$. We require that $l$ and $r$ are symbolic terms, rather than arbitrary terms:
\[l, r ::= x \mid \gamma(l_1,\, l_2,\, ...,\, l_n)\]
Thus, the \emph{symbolic} rewrite system is ignorant of the let-construct, of lambda abstractions, and of applications.

\paragraph*{Rewriting contexts}

For this work, we assume that all the rewrite relations (including those marked $\rw$, $\rwml$, $\succ$ and $\dotsucc$) are compatibly-closed, so they can be applied in any term context. This is a fairly standard assumption for rewrite systems.

In some application areas, one may wish to constrain the eligible rewrite contexts, for example to a call-by-value evaluation order. Doing so is normal in the study of computational effects. But by allowing rewrites anywhere, our result remains more general.

\section{The Recursive Path Ordering Defined}

Now we define the recursive path ordering $\succ$, which relates terms of the symbolic part of the language, and which is the key tool of the normalization proof.

The RPO is defined with respect to an ordering $>_\Sigma$ on rewritable symbols. We will write $>$ for the symbol ordering when it is clear from context. The relation $\succ$ is extended to a lexicographical ordering on a sequence of terms by writing $\succ_{lex}$. (The (HO)RPO usually allows symbols whose args are ordered either by a multiset ordering or a lexicographical ordering. Presently we only defines the lexicographical one.)
We use $\succeq$ for the union of $\succ$ with $=$.
\begin{definition}[RPO]
Define $s \succ t$ to hold when one of the following does:
\begin{enumerate}
\item $s = \gamma(s_1,\, \dots,\, s_n)$,
      $t = \gamma(t_1,\, \dots,\, t_n)$ and
      $\vec s \succ_{lex} \vec t$ and
      for all $j$, $s \succ t_j.$
\item $s = \gamma(s_1,\, \dots,\, s_m)$, $t = \gamma'(t_1,\, \dots,\, t_n)$ and $\gamma > \gamma'$ and for all $j$, $s \succ t_j.$
\item $s = \gamma(s_1,\, \dots,\, s_m)$ and for some $i$, $s_i \succeq t.$
\end{enumerate}
\end{definition}
It is easy to miss that this relation is inductively defined, and there is a base case hidden in case (3), in the $=$ part of the $\succeq$ relation. All derivations of the RPO end in leaves which are assertions of the right-hand term being equal to an immediate subterm of the left-hand term.

Let's take a moment to understand the purpose of this step intuitively, and where it fits in the larger proof. The $\succ$ relation essentially captures a large class of terminating rewrite systems that could be defined for a given effect-signature $\Sigma$ and an ordering among the symbols. The ordering will be specific to the particular rewrite system, but the $\succ$ relation abstracts slightly from the rewrite rules themselves. It is usually a superset of the relation $(\rw)$ of interest, so $(\rw) \subseteq (\succ)$. The user must also check their rewrite rules ($\rw$) do in fact meet the above criteria (qualifying as a $(\succ)$ relation), but this is often easy to do, and one then gets a big termination proof ``for free,'' as it were.

The $\succ$ relation abstracts only the symbol-rewriting rules; to extend it through the metalanguage, we define \[(\dotsucc) \triangleq (\succ) \cup (\rwml).\] And this $\dotsucc$ is the relation for which we will prove strong normalization. As a result, the target calculus where $(\rw) \cup (\rwml)$ is the relation of interest must also be strongly normalizing.

We can also present the RPO in terms of two powerful inference rules:
\begin{mathpar}
\inferrule[rpo-subterm]{s_i \succeq t \text{ for some $s_i$ } }
          {\gamma(s_1,\, \dots,\, s_n) \succ t}
\and
\inferrule[rpo-symbol]{(\gamma, \vec s) >_{RPO} (\gamma', \vec t) \\ s \succ t_i \text{ for each $t_i$}}
          {\gamma(s_1,\, \dots,\, s_n) \succ \gamma'(t_1,\, \dots,\, t_m)}
\end{mathpar}
Where the $>_{RPO}$ ordering is defined as a lexicographical ordering with the components $>_\Sigma$ and $\succ_{lex}$.

In what follows, a \emph{reduct} of $t$ is a term $t'$ for which $t\dotsucc t'$, and
we write $SN(t)$ if $t$ strongly normalizes under the relation $\dotsucc$. When a term is strongly normalizing, we can perform induction on its reduction tree (we only use this for the reduction tree under ($\dotsucc$), not the other term relations); to invoke this principle we will write ``induction on $t$, ordered by $\dotsucc$.'' When we have several normalizing terms handy, we might use simultaneous induction on all of them, where the proposition is assumed to hold for the group where any one is reduced.

\subsection{Continuations}

A key difficulty in the proof is showing strong normalization in the presence of the let-assoc rule, which reorganizes the term in a progress-making way, but does not make it smaller. Thus we need a 
construct to allow tracking and inducting on that progress.
In the let-assoc rule, $\kwlet x \be (\kwlet y \be s_1 \kwin s_2) \kwin s_3$ has a let on its inner left-hand side, and we want to see that go away. The reduct of this form, $\kwlet y \be s_1 \kwin \kwlet x \be s_2 \kwin s_3$, has the let form on its inner right-hand-side, which is closer to a normal form. We have decreased the number of let-forms that are in the left-hand-sides of other let-forms.
Therefore, following Lindley and Stark~\cite{lindley2005reducibility}, we define continuations, which are a stack of the ``wrong kind'' of let context---the $\kwlet x \be \hole \kwin s_3$ of the foregoing explanation. A let-assoc rewrite will reduce the size of this stack, so in inductive proofs we have a way to prove progress is being made.

Define continuations $K$ as follows:
\begin{eqnarray*}
K &::=& \epsilon \mid K \circ F \\
F &::=& \kwlet x \be \hole \kwin u
\end{eqnarray*}
Write $F[t]$ for the term that results from filling the hole in $F$ with $t$. Write $K@t$ for that which results from filling the holes recursively, so $\epsilon@t = t$ and $K\circ F@t = K@F[t]$. Write $|K|$ for the number of frames in $K$, so $|\epsilon| = 0$, $|K\circ F| = 1 + |K|$.

Write $K \dotsucc K'$ if $K@x \dotsucc K'@x$ for some fresh $x$. Note that for an individual frame, the only rewrites are of the form $\kwlet x \be \hole \kwin u \dotsucc \kwlet x \be \hole \kwin u'$ with $u \dotsucc u'$.

\section{Strong Normalization}

Following Lindley-Stark, define a type-indexed reducibility predicate as follows:
\begin{eqnarray*}
  \Red B(t) &\triangleq& SN(t) \\
  \Red {S\to T}(s) &\triangleq& \Red T(st) \hbox{ for all $t$ where $\Red S(t)$} \\
  \Red {\kw{E}(T)}(t) &\triangleq& SN(K@t) \hbox{ for all $K$ where $\Red T^\top(K)$} \\
  \Red {T}^\top(K) &\triangleq& SN(K@\effpure(t)) \hbox{ for all $t$ where $\Red T(t)$}
\end{eqnarray*}
We will write, equivalently, $t \in \Red T$ or $\Red T(t)$ as is convenient for the prose.

Define the set of \emph{neutral} terms to be these:
\begin{itemize}
\item An application, $st$.
\item A variable, $x$.
\end{itemize}

We have four standard properties of $\Red-$ with standard proofs.

\begin{lemma}\label{lem:red-inhab}
For any $T$, $\Red T$ is inhabited by some term.
\end{lemma}
\begin{proof}
By induction on $T$.
\begin{itemize}
  \case $B$. A term which is a free variable satisfies this.
  \case $S \to T'$. By IH we have a term $t \in \Red {T'}$ and so $\lambda x. t \in \Red {S \to T'}$.
  \case $\kw E(T')$. By IH we have a term $t \in \Red {T'}$ and so $pure(t) \in \Red {\kw E(T')}$. \qedhere
\end{itemize}
\end{proof}

\begin{lemma}\label{lem:red-sn}
For any $T$, $t$, if $t \in \Red T$ then $t$ strongly normalizes.
\end{lemma}
\begin{proof}
By induction on $T$ and appeal to the $\Red T$ definition. In the $\kw E(T')$ case, $t$ is a subterm of something directly asserted to be SN. In the $T_1 \to T_2$ case, the IH gives strong normalization of a term which has $t$ as a subterm. We need that $\Red S$ is inhabited, which Lemma~\ref{lem:red-inhab} shows.
\end{proof}

\begin{lemma}
For any $s \in \Red T$ with $s \dotsucc s'$, we have $s' \in \Red T$.
\end{lemma}
\begin{proof}
By induction on $T$.
\begin{itemize}
\case $B$. $s'$ merely needs to be SN, and it is by virtue of being a reduct of $s$.
\case $T_1 \to T_2$. To show that $s't \in \Red {T_2}$ for any $t\in\Red {S_1}$. We have that $st \in \Red{T_2}$. But $st \dotsucc s't$ so the conclusion follows from the IH.
\case $\kw E(T')$. To show that $K@s' \in SN$ for $K \in \Red T^\top$. Again, $K@s$ reduces to $K@s'$ and since the former is SN, the latter is too.
\qedhere
\end{itemize}
\end{proof}

\begin{lemma}
  Given a neutral $s$, if each of its reducts is in $\Red T$ then $s$ is in $\Red T$.
\end{lemma}
\begin{proof}
By induction on the structure of $T$.
\begin{itemize}
\case $B$ Since the reducts are in $\Red B$, they are in SN, and this satisfies the definition of $\Red B$.
\case $S\to T'$ To show that $st \in \Red {T'}$ for each $t\in \Red S$. We have that $t$ is SN; proceed by induction on the reduction tree of $t$. Examine reductions of $st$. Since $s$ is neutral, it is not a $\lambda$-abstraction, so there is no $\beta$-reduction at the head. The only reducts are $s't$ (where $s\dotsucc s'$) and $st'$ (where $t\dotsucc t'$). In the first case, the lemma hypothesis is sufficient. In the second case, the inner IH is sufficient.
\case $\kw{E}(T')$ Given $K \in \Red {T'}^\top$, we want to show $SN(K@s)$. By induction on $K$. Since $s$ is neutral, the only reducts are 
$K'@s$ (where $K \dotsucc K'$) and
$K@s'$ (where $s \dotsucc s'$).
(Note there is no meta\-language rule rewriting the frame $F$ into the application $st$ in $F[st]$ and one cannot be supplied by the symbol-rewrites.) In the first case, the inner IH is sufficient. In the second case, the lemma hypothesis is sufficient.\qedhere
\end{itemize}
\end{proof}

\begin{lemma} \label{lem:unsubst-red}
  If $SN(u\{t/x\})$ then $SN(u)$.
\end{lemma}
\begin{proof}
  Constructively, every reduction in the reduction tree of $u$ has an analogue in that of $u\{t/x\}$. As a consequence, the tree for $u$ can be no larger than that of the other term, and cannot be divergent when the latter is convergent.
\end{proof}

Now we show that each term-former can construct a reducible term, given appropriate conditions.

\begin{lemma}\label{lem:app-red}
  If $\Red {S\to T}(s)$ and $\Red S(t)$ then $\Red T (st)$.
\end{lemma}
\begin{proof}
  Immediate from the definition of $\Red {S\to T}$.
\end{proof}

\begin{lemma}\label{lem:lam-red}
  If $\Red T(u\{t/x\})$ for every $t$ in $\Red S$ then $\Red {S\to T}(\lambda x. u)$.
\end{lemma}
\begin{proof}
  Since $(\lambda x. u)t$ is neutral, it is sufficient to show that all its reducts are reducible.
  We have that $t$ is SN by virtue of being in $\Red S$. We have that $u$ is SN by Lemma~\ref{lem:unsubst-red}, and therefore we can apply simultaneous induction on the two rewrite trees. The inductive hypotheses are that $(\lambda x.u')t$ is reducible, for any $u \dotsucc u'$, and that $(\lambda x.u)t'$ is reducible, for any $t \dotsucc t'$.
  Now we take those cases on the reducts of $(\lambda x. u)t$.
  \begin{itemize}
    \case $(\lambda x.u) t \dotsucc (\lambda x.u') t$ for $u \dotsucc u'$; this is reducible by IH.
    \case $(\lambda x.u) t \dotsucc (\lambda x.u) t'$ for $t \dotsucc t'$; this is reducible by IH.
    \case $(\lambda x.u) t \dotsucc u\{t/x\}$; this is reducible by lemma hypothesis.\qedhere
  \end{itemize}
\end{proof}

\begin{lemma}\label{lem:K-dotsucc-conserves-length}
If $K \dotsucc K'$ then $|K| \geq |K'|$.
\end{lemma}
\begin{proof}
By structural induction on K. If $K = \epsilon$, there is no reduction. If $K = K_0 \circ F$, we have reductions $K_0 \dotsucc K_0'$ and $F \dotsucc F'$, which conserve length (the former by IH).

Via (let-assoc), we also have $K = K_0 \circ F_1 \circ (\kwlet x \be \hole \kwin u) \circ K_1$ and $K' = K_0 \circ (\kwlet x \be \hole \kwin F_1[u]) \circ K_1$. And this is one frame shorter.
\end{proof}

\begin{lemma}\label{lem:pure-red}
If $t \in \Red T$ then $\effpure(t) \in \Red {\kw{E}(T)}$.
\end{lemma}
\begin{proof}
To show: that $K@\effpure(t) \in \Red {\kw{E}(T)}$ for any $K \in \Red T^\top$.
But this is immediate from the definition of $\Red T^\top$.
\end{proof}

\begin{lemma}\label{lem:let-simple-red}
If $s \in SN$ and $K@(u\{s/x\}) \in SN$, then $K@(\kwlet x \be pure(s) \kwin u) \in SN$.
\end{lemma}
\begin{proof}
By induction on $(|K|,\, (s,\, u,\, K))$ ordered by $(>,\, (\dotsucc_{lex}))$. Proceed by showing all reducts of $K@\kwlet x \be pure(s) \kwin u$ are in $SN$.
\begin{itemize}
\case $K'@\kwlet x \be pure(s) \kwin F[u]$ where $K'\circ F = K$, by let-assoc. To apply the IH, we need to show that $K$ and $F[u]$ meet the lemma premises, that $K'@(F[u]\{s/x\}) \in SN$. Note $K'@F[u] = K@u$ and $x$ cannot be free in $F$, by the let-assoc side condition, therefore $K'@(F[u]\{s/x\}) = K'@(F[u\{s/x\}]])$.) Furthermore $|K'| < |K|$, so the metric decreases.
\case The reduct is $K@u\{s/x\}$. By hypothesis.
\case The reduct is $K@\kwlet x \be pure(s') \kwin u$ where $s \dotsucc s'$. By IH.
\case The reduct is $K@\kwlet x \be pure(s) \kwin u'$ where $u \dotsucc u'$. By IH.
\case The reduct is $K'@\kwlet x \be pure(s) \kwin u$ where $K \dotsucc K'$. By IH. \qedhere
\end{itemize}
\end{proof}

\begin{lemma}\label{lem:let-red}
If $s \in \Red {\kw{E}(S)}$ and $u$ is such that for all $s' \in \Red {S}$ we have $u\{s'/x\} \in \Red {\kw{E}(T)}$, then $\kwlet x \be s \kwin u \in \Red {\kw{E}(T)}$.
\end{lemma}
\begin{proof}
We show that $SN(K@\kwlet x \be s \kwin u)$ for any $K \in \Red T^\top$. First, we show $K' = K\circ(\kwlet x \be \hole \kwin u) \in \Red S^\top$, which in other words says that $SN(K@\kwlet x \be pure(s'') \kwin u)$ for any $s'' \in \Red S$. This we get from Lemma~\ref{lem:let-simple-red} (by hypothesis, $u\{s''/x\}\in \Red T$ and further $K'@(u\{s''/x\}) \in SN$ as required by Lemma~\ref{lem:let-simple-red}). Now it follows by the definition of $s \in \Red S$ that $K'@s = K@\kwlet x \be s \kwin u \in SN$.
\end{proof}
%
%

The next lemma shows a property of the undotted $\succ$, that is, the raw RPO relation, which will be used as a subroutine in some inductive proofs to follow.
Since this is an extraction of an inductive step, it is stated in terms of a ``lemma hypothesis'' which will align with some outer induction hypothesis in the cases where it is used.

To make the lemma appropriately general, we define \emph{contexts} to encompass the various kinds of settings in which terms can be placed to prove reducibility:
\begin{eqnarray*}
C &::=& K \bigm\vert \hole t \bigm\vert \hole
\end{eqnarray*}
And write $C[s]$ to denote filling the context with a term:
\begin{eqnarray*}
C[s] = 
\begin{cases*}
K@s & when $C = K$ \\
st & when $C = \hole t$ \\
s & when $C = \hole$
\end{cases*}
\end{eqnarray*}

\begin{lemma}[RPO step]\label{lem:rpo-step}
Given some context $C$, and $\gamma(s_1,\, \dots,\, s_n) = s \succ t$, with each $C[s_i] \in SN$, and a ``lemma hypothesis'' that
\begin{itemize}
\item Given any $(\gamma'$,\, $\vec t)$ having $(\gamma,\, \vec s)$ greater than $(\gamma',\, \vec t)$ under the lexicographic ordering $((>_\Sigma),\, (\dotsucc_{lex}))$, we have $C[\gamma'(t_1,\, \dots,\, t_n)] \in SN$,
\end{itemize}
then $C[t] \in SN$.
\end{lemma}
\begin{proof}
    We show that $C[t]$ is SN by induction on the size of $t$, and take cases on the RPO rule that proves $s \succ t$:
    \begin{itemize}
      \case (1) $t = \gamma(t_1,\, \dots,\, t_n)$.
      First we show $C[t_i]$ is SN, which is given by the induction hypothesis (noting $t_i$ is smaller than $t$).
      Then by the lemma hypothesis, $C[t] = C[\gamma'(t_1,\, \dots,\, t_n)]$ is SN. Satisfying the ordering required by the lemma hypothesis, $\gamma$ has not changed and the RPO rule has offered $\vec s \succ_{lex} \vec t$, in turn implying $\vec s \dotsucc_{lex} \vec t$. Also note that the lemma hypothesis itself is preserved, when inducting.
      \case (2) $t = \gamma'(t_1,\, \dots,\, t_m)$ and $\gamma >_\Sigma \gamma'$.
      First we show $C[t_i]$ is SN, which is given by the induction hypothesis (noting $t_i$ is smaller than $t$).
      Then by the lemma hypothesis, $C[t] = C[\gamma'(t_1,\,\dots,\, t_m)]$ is SN. The lemma hypothesis is satisfied by $\gamma >_\Sigma \gamma'$.
      \case (3) $s_i \succeq t$. Here $C[t]$ is in the reduction tree of $C[s_i]$, which was assumed SN, so then $C[t] \in SN$.
       \qedhere
    \end{itemize}
\end{proof}

\begin{lemma}\label{lem:symbol-sn}
  Let $s = \gamma(s_1,\, \dots,\, s_n)$.
  Given a context $C$, if each $s_i$ has $C[s_i] \in SN$, then $C[s] \in SN$.
\end{lemma}
\begin{proof}
By cases on the type of $s$. In each case, we study the possible terms $C[s]$ and show that their reducts are all SN, and thus that $C[s]$ is.
\begin{itemize}
\case $B$. It must be that $C = \hole$ so we just study the term $s$ itself.
  By induction on the tuple \( (\gamma,\, \vec s) \) lexicographically ordered by
  \[((>_\Sigma),\, (\dotsucc_{lex})).\] By cases on the reducts of $s$:
  \begin{itemize}
    \case $\gamma(s_1,\, \dots,\, s_i',\, \dots,\, s_n)$ for some index $i$ and $s_i \dotsucc s_i'$. The inner inductive hypothesis applies because $\gamma$ is unchanged, while $\vec s$ has decreased under $\dotsucc_{lex}$.
    \case $t$ where $s \succ t$. (Under the un-dotted $\succ$ relation.)
    Lemma~\ref{lem:rpo-step} applies. Our inner induction hypothesis implies the ``lemma hypothesis'' of Lemma~\ref{lem:rpo-step}.
  \end{itemize}
  
\case $T_1 \to T_2$.
  The context $C$ is either $\hole$ or $\hole s'$. Showing that $s s'$ is SN covers all the cases for $s$ alone, so we only show those.
  Proceed by lexicographical induction on \((s',\, \gamma,\, \vec s)\),
  ordered by \(( \dotsucc,\, >_\Sigma,\, \dotsucc_{lex})\).
  \begin{itemize}
    \case a subterm reduces: either some $s_i \dotsucc s'_i$ or the applicand $s' \dotsucc s''$. By IH, with reduction sequence decreasing.
    \case $s \succ t$. 
    Lemma~\ref{lem:rpo-step} applies. Our induction hypothesis implies the ``lemma hypothesis''.
  \end{itemize}

\case $\kw{E}(T')$.
  Here $C = K$. Since $K@s_i \in SN$, we know that $K$ itself is SN.
  We proceed by an lexicographic induction on the tuple \[
      (|K|,\, K, \gamma,\, \vec s)
      \text{ ordered by }
      (>,\, \dotsucc,\, >_\Sigma,\, \dotsucc_{lex}).
  \]
  We show that every reduct of $K@s = K@\gamma(s_1,\, \dots,\, s_n)$ is strongly-normalizing, and thus that the term itself is. By cases on those reducts:
  \begin{itemize}
    \case $K@\gamma(s_1,\, \dots,\, s_i',\, \dots,\, s_n)$ for some index $i$ and $s_i \dotsucc s_i'$. The IH applies: $K$ and $\gamma$ are unchanged and the arguments $\vec s$ have lexicographically reduced under $\dotsucc$.
    \case $K'@s$ where $K \dotsucc K'$. The IH applies because the continuation has gotten no longer (Lemma~\ref{lem:K-dotsucc-conserves-length}) and $K'$ is in the reduction tree of $K$.
    \case $K'@\gamma(F[s_1],\, \dots,\, F[s_n])$ where $K = K' \circ F$ and $\gamma$ is an effect symbol (eff-assoc).
    The IH applies because $K$ has gotten shorter (as $K'$). We also need that $K'@F[s_i]$ is SN, to satisfy the IH, but $K'@F[s_i] = K@s_i$, which we already know is SN.
    \case $K@t$ where $s \succ t$. (Under the un-dotted $\succ$ relation.)
    Lem\-ma~\ref{lem:rpo-step} applies.
    \qedhere
  \end{itemize}
\end{itemize}
\end{proof}

\begin{lemma}\label{lem:symbol-red}
  For a rewritable symbol $\gamma : S_1 \times \cdots \times S_n \to T$,
  if each $s_i \in \Red {S_i}$, then $s = \gamma(s_1,\, \dots,\, s_n) \in \Red {T}$.
\end{lemma}
\begin{proof}
By cases on the type of $s$.
\begin{itemize}
\case $B$. To show $s \in \Red B$ for which we only need that $s \in SN$, and we get this from Lemma~\ref{lem:symbol-sn}.
\case $T_1 \to T_2$. To show $s s' \in \Red {T_2}$ for any $s' \in \Red {T_1}.$
  Because $s s'$ is neutral, we can show just that all the reducts of $s s'$ are reducible.
  By induction on $((\gamma, \vec s), s')$ ordered by $((>_\Sigma, \dotsucc_{lex}), \dotsucc_{lex})_{lex}$.
  The only reductions of $s s'$ are in $s$ or in $s'$. If in $s'$, the IH suffices.
  If in $s$, there are two possibilities:
  \begin{itemize}
  \case the reduction is of the form $s = \gamma(s_1,\, \dots s_i,\, \dots,\, s_n) \dotsucc \gamma(s_1,\, \dots s_i',\, \dots,\, s_n)$ with $s_i \dotsucc s_i'$, in which case the IH suffices.
  \case $s \succ t$.
  If it reduces by $s \succ t = \gamma'(t_1,\,\dots,\,t_m)$ with $(\gamma, \vec s) >_{RPO} (\gamma', \vec t)$, then IH proves the point. On the other hand, if it is $s_i \succeq t$ then $t \in \Red {S_i}$ by virtue of the $s_i \in \Red {S_i}$ assumption.
  \end{itemize}

\case $\kw{E}(T')$. To show $K@s \in SN$ for any $K \in \Red {T'}^\top$. Because $s_i \in \Red {S_i}$, we have $K@s_i \in SN$, which satisfies the premises of Lemma~\ref{lem:symbol-sn}, thus $K@s\in SN$ as needed.
\qedhere
\end{itemize}

\end{proof}

Write $s\{\vec s/\vec x\}$ for the operation of simultaneously substituting each $s_i$ for the free variable $x_i$ within $s$: $s\{\vec s/\vec x\} = s\{s_1/x_1,\,\dots,\,s_n/x_n\}$.

\begin{lemma}[Reducibility]
Given $x_1 : S_1,\,\dots,\,x_n:S_n \vdash t : T$, for all $\vec s \in \Red {\vec S}$, we have $t\{\vec s/\vec x\} \in \Red T$.
\end{lemma}
\begin{proof}
By structural induction on $t$.
\begin{itemize}
\case $x_i$. Then $S_i = T$. Now $s_i \in \Red {S_i}$ and $x_i\{s_i/x_i\} = s_i \in \Red {S_i} = \Red T$.
\case $s't'$. Immediate from Lemma~\ref{lem:app-red} and the IH.
\case $\lambda x. u$. The type derivation has $\Gamma,\, x:S \vdash u : T'$. Let $s'$ be in $\Red S$. By inductive hypothesis, using $\Gamma$ extended with $s'$, we have $u\{\vec s/\vec x, s'/x\} = u\{\vec s/\vec x\}\{s'/x\} \in \Red T$, and thence by Lemma~\ref{lem:lam-red}, $\lambda x.u \in \Red {S\to T'}$.
\case $\effpure(M)$. Immediate from Lemma~\ref{lem:pure-red} and the IH.
\case $\gamma(t_1, t_2, ..., t_n)$. Immediate from Lemma~\ref{lem:symbol-red} and the IH.
\case $\kwlet x \be t' \kwin u$. 
The type derivation is such that $\Gamma \vdash t' : S$ and $\Gamma, x:S \vdash u : T$.
By IH, we have $t'\{\vec s/\vec x\} \in \Red S$ and then $u\{t'/x\}\{\vec s/\vec x\}  = u\{\vec s/\vec x, t'\{\vec s/\vec x\}/x\} \in \Red T$ and from there, Lemma~\ref{lem:symbol-red}.
\qedhere
\end{itemize}
\end{proof}

\section{Revisiting the Examples}

So what does all this give us? Can we use this technique to show termination for some interesting calculi?

\paragraph*{Global state}

Recall the rewrite rules of global state given earlier:
\begin{align*}
  assign_i(get(t_1,\, \dots,\, t_n)) &\rw assign_i(t_i)\\
  assign_i(assign_j(s)) &\rw assign_j(s) \\
  get(t_1,\, \dots,\, t_i,\, \dots,\, t_n) &\rw get(t_1,\, \dots,\, s_i,\, \dots,\, t_n)\\
&\text{
\hspace{2em}
where $t_i = get(s_1,\, \dots,\, s_n)$}
\end{align*}
These are easily shown to be normalizing: each only needs the RPO(3) case once or twice.

\paragraph*{Nondeterminism}
\[or(or(s, t), u)) \rw or(s, or(t, u))
\]
With one effect symbol, the $>_\Sigma$ relation is empty. But we check the RPO conditions on the solitary rule. Because the function symbols match, we use RPO(1). We have to check the lexicographical ordering of the arguments: $or(s,\, t) \succ s$ (by RPO(3)) so we don't need to check the second argument. Then each argument on the RHS must be less than the whole term on the left, which is easily done (recursively applying RPO(3) or (1)). Note that the lexicographical ordering was crucial to ensuring this rule makes progress toward termination.

\paragraph*{Parallelism}
For parallelism, set $par >_\Sigma e$ for each other effect symbol $e$. Recursive comparisons can then be carried out. Note for example that $par(e(s_1,\dots,s_n),\, t) \succ par(s_i,\, t)$ under RPO(1), the arguments are lexicographically decreasing, and $s_i$ and $t$ can each be found as subterms.

\paragraph*{Request-retry}

Restating the rewrite rules:
\begin{eqnarray*}
retry(zero(), request(t, s_1,\, \dots,\, s_n)) &\rw& t \\
retry(succ(u), request(t, s_1,\, \dots,\, s_n)) &\rw& request(retry(u, t'), s_1,\, \dots,\, s_n)\\
&&&where $t' = request(t, s_1,\, \dots,\, s_n)$
\end{eqnarray*}
Let $\id{retry} >_\Sigma \id{request}.$
The first rule has $t$ as a subterm of the left-hand side, so it is in RPO.
The second rule has $\id{retry} >_\Sigma \id{request}$ and then we need to show
\[retry(succ(u), request(t, s_1,\, \dots,\, s_n)) \succ retry(u, t')\]
and each $s_i$ has
\[retry(succ(u), request(t, s_1,\, \dots,\, s_n)) \succ s_i.\]
The latter is easy, via the subterm rule.
For the former, the head symbol matches, and then the immediate subterms $u$ and $t'$ can both be found as subterms of the left-hand side.

\section{Related Work}

Johann, et al.~\cite{johann2010generic} give ``a generic operational metatheory for algebraic effects''. The authors work with computation trees, or traces, like those which are the normal forms of algebraic-effect systems in the absence of equations. An equivalence (in fact, a preorder), between computation trees is given for each kind of effect system, but it is given through a separate definition which simulates the operation of each effect on a separate state-representation (a kind of abstract machine). By contrast, we have explored what happens when the effects can be defined by rewrites on the effect symbols themselves.
Gavazzo and Faggian~\cite{gavazzo2021relational} explain monadic effects in rewrite systems. This work interprets the rewrite relation \emph{itself} as monadic/effectful, so for example the rewrite relation can have a probabilistic distribution on its possible right-hand terms.

We have set out to show that a set of rewrite rules can be applied in any order and still normalize. But one may instead choose a \emph{particular} rewriting or normalization strategy. Normalization By Evaluation is one such approach, and Ahman and Staton~\cite{ahman2013normalization} have shown how to perform NBE on a calculus with algebraic effects and a sequencing form (like our $\kw{let}$).

We build on the long history of rewrite-rule orderings to prove termination. Dershowitz~\cite{dershowitz1982orderings} gives the original RPO ordering and termination proof. That proof is entirely different from the reducibility method, which is necessitated by the difficulties of the let-assoc rule. 
Okada~\cite{okada1989strong} is the hero of our present work, as it is the first paper to show a general proof for strong normalization of any SN rewrite system crossed with the syntax of simply-typed lambda-calculus. Showing such orthogonality between rewrites and other syntax features is the spirit of the present work.

\section{Future Work}
So far, we have only shown a modest improvement on existing strong-normalization systems. The symbols in the system have their own rewrite rules, but are allowed to interact with $\kw{let}$ in just one way, commuting out of the subject position. We hope to give similar strong-normalization proofs for systems in which the rewrite system can specify further interactions with $\kw{let}$ (although some restrictions may remain). Our grand ``test cases'' for the technique are the systems in Cooper~\cite{cooper2009script} and Ricciotti and Cheney~\cite{ricciotti2020snhorq}: when we can prove these strongly-normalizing with only a symbol-ordering, we will have succeeded.

\paragraph*{Thanks}
Thanks to Sam Lindley, Matija Pretnar, and Wilmer Ricciotti for helpful comments shaping this work.

\bibliography{horpo-algeff-rewr}

\end{document}